\documentclass[english]{amsart}
\usepackage{lmodern}

\usepackage[T1]{fontenc}
\usepackage[latin9]{inputenc}
\usepackage{listings}
\usepackage{color}
\definecolor{shadecolor}{rgb}{1, 0, 0}
\usepackage{framed}
\usepackage{ifthen}
\usepackage{amssymb}
\usepackage{amsmath}
\usepackage{natbib}

\usepackage[unicode=true,
bookmarks=false,
breaklinks=true,pdfborder={0 0 0},backref=false,colorlinks=false]
{hyperref}
\hypersetup{
	pdftitle={NOOP: A Domain-Theoretic Model Of Nominally-Typed Object-Oriented Programming},
	pdfauthor={Moez AbdelGawad},
	pdfauthor={Robert Cartwright},
	pdfsubject={Nominal-Typing in OOP},
	pdfkeywords={Nominal-Typing, OOP, NOOP, Structural-Typing, Java, C\#, C++, Scala, Inheritance, Subtyping}}

\makeatletter

\newcommand{\noun}[1]{\textsc{#1}}
\renewcommand{\cite}[1]{\citep{#1}}

\newenvironment{lyxcode}
{\par\begin{list}{}{
\setlength{\rightmargin}{\leftmargin}
\setlength{\listparindent}{0pt}
\raggedright
\setlength{\itemsep}{0pt}
\setlength{\parsep}{0pt}
\normalfont\ttfamily}%
 \item[]}
{\end{list}}
\newcommand{\code}[1]{\texttt{#1}}
\newtheorem{theorem}{Theorem}[section]
\newtheorem{lemma}{Lemma}[section]
\newtheorem{definition}{Definition}[section]
\newtheorem{proposition}{Proposition}[section]

\makeatother

\usepackage{babel}

\global\long\def\NOOP{\mathbf{NOOP}}
\global\long\def\SOOP{\mathbf{SOOP}}
\global\long\def\rec{\multimap}
\global\long\def\dom#1{\mathcal{#1}}
\global\long\def\strfunarr{\multimap\!\rightarrow}
\global\long\def\Sig#1{\mathsf{s_{#1}}}
\global\long\def\subsign{\trianglelefteq}
\global\long\def\ext{\blacktriangleleft}

\begin{document}
	
\title[$\NOOP$: A Domain-Theoretic Model Of Nominally-Typed OOP]{$\NOOP$: A Domain-Theoretic Model Of Nominally-Typed Object-Oriented Programming}


\author[M. AbdelGawad]
{Moez AbdelGawad\smallskip\\
College of Mathematics and Econometrics, Hunan University\\
Changsha 410082, Hunan, P.R. China\smallskip\\
\code{moez@cs.rice.edu}\\}

\author[R. Cartwright]
{\\Robert Cartwright\smallskip\\
	Computer Science Department, Rice University\\
	Houston, Texas 77005, USA\smallskip\\
	\code{cork@cs.rice.edu}}

\date{July 2016}

\begin{abstract} 
The majority of industrial-strength object-oriented (OO) software
is written using nominally-typed OO programming languages. Extant
domain-theoretic models of OOP developed to analyze OO type systems
miss, however, a crucial feature of these mainstream OO languages: nominality.
This paper presents the construction of $\NOOP$ as the first domain-theoretic
model of OOP that includes full class/type names information
found in nominally-typed OOP. Inclusion of nominal information
in objects of $\NOOP$ and asserting that type inheritance in statically-typed OO programming
languages is an inherently nominal notion allow readily proving that
type inheritance and subtyping are completely identified in these languages.
This conclusion is in full agreement with intuitions of developers
and language designers of these OO languages, and contrary to the belief that \textquotedblleft{}inheritance
is not subtyping,\textquotedblright{} which came from assuming non-nominal
(\emph{a.k.a.}, structural) models of OOP. 

To motivate the construction of $\NOOP$, this paper briefly presents the
benefits of nominal-typing to mainstream OO developers and OO language designers,
as compared to structural-typing. After presenting $\NOOP$, the paper further
briefly compares $\NOOP$ to the most widely known domain-theoretic models of OOP.
Leveraging the development of $\NOOP$, the comparisons presented in this paper
provide clear, brief and precise technical and mathematical accounts for the relation
between nominal and structural OO type systems. $\NOOP$, thus, provides a firmer
semantic foundation for analyzing and progressing nominally-typed
OO programming languages.
\end{abstract}

\maketitle


\section{\label{sec:Intro}Introduction}

To evolve and improve the type systems of mainstream object-oriented programming languages
such as Java~\cite{JLS14}, C\#~\citeyearpar{CSharp2015}, C++~\citeyearpar{CPP2011}, and
Scala~\cite{Odersky14}, which utilize class name information in defining object types
and OO subtyping, a precise mathematical model of these languages is needed. A precise model
of nominally-typed OOP allows accurate reasoning and analysis of these mainstream
OO programming languages. Imprecise models, on the other hand, lead to inaccurate conclusions.

An object in nominally-typed OO languages is associated with its class\footnote{The term `type' is
overloaded. In this paper, the term has mainly two related but distinct meanings. The first
meaning, mainly used by OO software developers, is a syntactic one, that directly translates to
the expression `class, interface, or trait' (in OO programming languages that support these constructs). 
In this sense, each class, interface, or trait \emph{is} a type.  The second meaning for
`type,' mainly used by mathematicians and programming languages researchers, is a semantic
meaning referring to the set of instances of a corresponding class/interface/trait. In this
sense, each class, interface, or trait \emph{corresponds} to a type.  Usually the context
makes clear which sense of the two is meant, but, to emphasize, sometimes we use the term
`class' for the syntactic meaning. As such, unless otherwise noted the term `class' in this
paper should be translated in the mind of the reader to `class, interface or trait.'}
name and the class names of its superclasses, as part of the meaning
of the object. Class names, in turn, are associated with class contracts, which are usually expressed, informally, in code documentation. Class contracts are thus implicitly encoded in class names. %

In nominally-typed OOP, two objects with the same structure but that have
different class name information are different objects, and they have
different types. The different class name information inside the two objects 
implies the two objects maintain different class contracts, and thus that the objects
are behaviorally dissimilar. The two objects are thus considered
semantically unequal. Further, in nominally-typed OO languages---where types and the
subtyping relation make use of class names and of the explicitly-specified type inheritance
relation between classes---instances of two classes that are not in the inheritance hierarchy
may not be replaced by each other (\emph{i.e.}, are not `assignment-compatible')
since they may not offer the degree of behavioral substitutability
intended by developers of the two classes.

Despite its clear semantic importance, class name information (henceforth,
`nominal information') that is embedded inside objects of many mainstream
OO programming languages is not included in the most recognized denotational
models of OOP that exist today. Models of OOP that lack nominal information
of mainstream OO languages are structural models of OOP, not nominal
ones. Examples of structurally-typed OO languages include O'Caml~\cite{OCamlWebsite}
(see~\cite{MacQueenMLOO02}) and research languages such as Modula-3~\cite{Cardelli89modula},
Moby~\cite{Fisher1999}, Strongtalk~\cite{Bracha1993}, and PolyTOIL~\cite{Bruce2003}. 
Structural models of OOP have led PL researchers to make some conclusions
about OOP that contradict the intuitions of the majority of mainstream
OO developers and language designers. For example, the agreement of type inheritance, at
the syntactic (\emph{i.e.,} program code) level, and subtyping, at
the semantic (\emph{i.e.}, program meaning) level, is a fundamental
intuition of OO developers using nominally-typed OO languages. However,
extant denotational models of OOP led to the inaccurate conclusion that ``inheritance
is not subtyping.''

Type inheritance, in class-based mainstream OO languages, is an inherently
nominal notion, due to the informal association of class names with
\emph{inherited} class contracts. Hence the discrepancy between conclusions
regarding inheritance that are based on a structural view of OOP and
the intuitions of the majority of mainstream OO developers, who adopt
a nominal view of OOP. This discrepancy motivated considering the
inclusion of nominal information in mathematical models of OOP.

This paper presents the construction of a mathematical model of OOP,
called $\NOOP$, that includes full nominal information of mainstream OO
programming languages. $\NOOP$ was first presented
in~\cite{NOOP} and its construction was summarized in~\cite{NOOPsumm}.

Having a model of OOP that includes
nominal information of nominally-typed OOP should enable progress
in the design of type systems of current and future mainstream OO
languages. Some features of the type systems of these languages
(\emph{e.g.}, generics)
crucially depend on nominal information. Accurately
understanding and analyzing these features, for the purposes of extending
the languages or designing new languages that include them, has proven
to be hard when using operational models of OOP or using structural denotational
models of OOP, which lack nominal information found in nominally-typed
OO languages. 
 Having a nominal
domain-theoretic model of OOP should make the analysis of features
of these languages that depend on nominal information easier and more
accurate. From the point of view of OO software development, having
better mainstream OO languages should result in greater productivity
for software developers and in them producing robust high-quality
software.

This paper is organized as follows. Section~\ref{sec:Related-Research}
presents a list of research related to this paper. Section~\ref{sec:WhyNom}
presents in brief the value of nominal typing to mainstream OO developers.
Section~\ref{sec:Rec} then starts the formal presentation of $\NOOP$
by presenting a new records domain constructor, called `rec,' that
is used in constructing $\NOOP$. Section~\ref{sec:Signatures} presents
class signatures and other related signature constructs, which are
syntactic constructs used to embody the nominal information found
in nominally-typed OOP. Section~\ref{sec:NOOP} presents the construction
of $\NOOP$, using `rec' and signature constructs, then it
presents a proof of the identification of inheritance and subtyping
in nominally-typed OOP. Section~\ref{sec:NOOPvsSOOP} then presents in brief
a comparison of $\NOOP$ to the most well-known domain-theoretic models
of OOP, namely the two structural models developed by Cardelli and by Cook.
Section~\ref{sec:Conclusions} presents the main conclusions we reached
based on developing $\NOOP$ and on comparing it to other domain-theoretic models
of OOP. Section~\ref{sec:Future-Work} concludes this paper by presenting
further research that can be developed based on $\NOOP$.

\section{\label{sec:Related-Research}Related Research}

$\NOOP$ is a domain-theoretic model of nominally-typed OOP. Dana Scott invented
and developed---with others including Gordon Plotkin---the
fields of domain theory and denotational semantics \citep[\emph{e.g.}, see~][]{DTAL,Stoy77,CatRecDomEqs82,Scott82,GunterHandbook90,Gierz2003,DomTheoryIntro}.
The development of denotational semantics has been motivated by researching
the semantics of functional programming languages such as Lisp~\cite{McCarthyBasis63,McCarthy96} and ML~\cite{ML-LCF78,Milner97}.

Research on the semantics of OOP has taken place subsequently. Cardelli built the
first widely known denotational model of OOP~\cite{Cardelli84,Cardelli88}.
Cardelli's work was pioneering, and naturally, given the research on modeling
functional programming extant at that time, the model Cardelli constructed
was a structural denotational model of OOP that lacked nominal information.\footnote{Significantly, Cardelli in
fact also hinted at looking for investigating nominal typing (on page~2 of~\cite{Cardelli87}).
Cardelli's hint, unfortunately, went largely ignored for years.}
Cook and his colleagues built on Cardelli's work to separate
the notions of inheritance and subtyping~\cite{CookDenotational89,Cook1989,CookInheritance90}. %
Later, other researchers (such as ~\cite{BruceFoundations02} and ~\cite{SimonsTheory02})
promoted Cardelli and Cook's structural view of OOP, and promoted
conclusions based on this view.

Martin Abadi, with Luca Cardelli, later presented \emph{operational} models
of OOP~\cite{SemObjTypes94,TheoryOfObjects95}. These models also 
had a structural view of OOP. Operational models with a nominal view of OOP 
got later developed however. In their seminal
work, Atsushi Igarashi, Benjamin Pierce, and Philip Wadler presented
Featherweight Java (FJ)~\cite{FJ/FGJ} as an operational model of
a nominally-typed OO language. Even though not the first operational
model of nominally-typed OOP (for example, see~\cite{drossopoulou99},~\cite{nipkow98}
 and~\cite{flatt98,flatt99}), FJ is the most widely-known operational
model of (a tiny core subset of) a nominally-typed OO language, namely
Java.\footnote{It is worthy to mention that $\NOOP$---as a more foundational
domain-theoretic model of nominally-typed OO languages (including Java)---provides
a denotational justification for the inclusion of nominal information in
Featherweight Java.}%


Other research that is similar to one presented here, but that had
different research interests and goals, is that of Reus and Streicher~\cite{Reus02,Reus02a,Reus03}.
In~\cite{Reus03}, an untyped denotational model of class-based OOP
is developed. Type information is largely ignored in this work (object
methods and fields have no type signatures) and some nominal information
is included with objects only to analyze OO dynamic dispatch. The
model of~\cite{Reus03} was developed to analyze mutation and imperative
features of OO languages and for developing specifications of OO software
and the verification of its properties.
 Analyzing the differences
between structurally-typed and nominally-typed OO type systems was
not a goal of Reus and Streicher's research, and in their work the
identification of inheritance and subtyping was, again (as in FJ), assumed rather
than proven as a consequence of nominality and nominal typing.

\section{\label{sec:WhyNom}The Value of Nominal-Typing in OOP}

In this section we briefly present the value of nominal-typing and nominal-subtyping to
OO software developers and OO language designers. More
details on the value of nominal-typing and nominal-subtyping can be found in~\cite{AbdelGawad2015}.

As hinted to in the Introduction (Section~\ref{sec:Intro}), the main semantic value of nominal-typing to mainstream OOP lies in the association of type
(\emph{i.e.,} class/interface/trait) names with behavioral contracts that are part of the public
interface of objects, making typing and subtyping in nominally-typed OO languages closer
to semantic typing and semantic subtyping than structural-typing and structural-subtyping are. Designing their software based on having public behavioral contracts allows OO developers to design robust software~\cite{bloch08}.

The semantic value of nominal type information leads nominally-typed and structurally-typed
OO languages to have different views of type names, where type names in nominally-typed OOP have \emph{fixed} meanings
(tied to the public contracts) while in structurally-typed OOP (in agreement with the tradition in functional programming) type names are viewed as mere `shortcuts for type expressions' that \emph{can} thus change their meanings, \emph{e.g.}, upon inheritance.  This difference in viewing type names leads
OO developers using structurally-typed OO languages to face problems---such as spurious
subtyping, missing subsumption, and spurious binary methods (see~\cite{AbdelGawad2015})---that are not found in nominally-typed OO languages.

Further, the identification of type inheritance with OO subtyping (`inheritance is subtyping')
resulting from nominal-typing (which we prove in this paper) enables nominally-typed OO languages
to present OO developers
with a simple conceptual model during the OO software design process.

Finally, due to the ubiquity of the need for objects in OOP to be ``autognostic''
(self-aware, \emph{i.e.} recursive, see~\cite{cook-revisited}) and given that recursive
data values can be typed using recursive types~\cite{MPS}, the ease by which recursive
types can be expressed in nominally-typed OO languages is a decided benefit
nominal-typing offers to OO software developers and designers~\cite{TAPL}. More details on the benefits of nominal-typing can be found in~\cite{AbdelGawad2015}.\\

Without further ado, we now start the presentation of $\NOOP$ as a model of
OOP that includes full nominal information found in many mainstream OO languages.

\section{\label{sec:Rec}`Rec' ($\rec$), A New Records Domain Constructor}

For the purpose of constructing $\NOOP$, we introduce a new domain constructor.
 In addition to $\NOOP$
including nominal information of mainstream OOP, $\NOOP$ models records
as \emph{tagged finite functions} rather than infinite functions,
as another improvement over extant domain-theoretic models of OOP
(particularly that of Cardelli and 
other models built directly on top of it, such as Cook's.)

Due to the finiteness of the shape of an object (the shape of an object is the set of names/labels of its fields and
methods), and due to the flatness of the domain of labels when labels are
formulated as members of a computational domain, modeling objects
in $\NOOP$ motivates defining a new domain constructor that is similar
to but somewhat different from conventional functional domain constructors.
This domain constructor, $\rec$, called `rec,' constructs tagged
finite functions, which we call \emph{record functions}. Record
functions are explicitly finite mathematical objects.

A domain $\dom R=\dom L\rec\dom D$, constructed using $\rec$, is the
domain of record functions modeling records with labels from a flat
domain $\dom L$ of labels to an arbitrary domain $\dom D$ of values.
Below we present the records domain constructor, $\multimap$,
then we discuss its mathematical properties. The definition of $\rec$
makes use of standard definitions of basic domain theory (See, for example,~\cite{DomTheoryIntro}. A
summary of domain theory notions used to construct $\NOOP$ is presented in~\cite{DomThSummCOOP14} and in Appendix~A of~\cite{NOOP,NOOPbook}.)

\subsection{Record Functions}

A record can be viewed as a finite mapping from a set of labels (as
member names) to fields or methods. Thus, we model records using explicitly
finite record functions. A \emph{record function} is a finite function
paired with a tag representing the input domain of the function\emph{.
}The tag of a record function modeling a record represents the set
of labels of the record. In agreement with the definition of shapes
of objects, we similarly call the set of labels of a record the \emph{shape}
of the record. The tag of a record function thus tells the shape of
the record.

\subsection{\label{sec:Definition-of-Records}Definition of $\multimap$}

Let $\dom L$ be the flat domain containing all record labels plus
an extra improper bottom label, $\bot_{\dom L}$, that makes $\dom L$
be a domain. (All computational domains must have a bottom element.)
Let $\dom D$ be an arbitrary domain, with approximation ordering
$\sqsubseteq_{\dom D}$ and bottom element $\bot_{\dom D}$. Domain
$\dom D$ contains the values that members of records are mapped to.

Let $\Subset$ denote the subdomain relation (see Definition 6.2 in~\cite{DomTheoryIntro}.)
If we let $\dom L_{f}$ range over arbitrary finite subdomains of
$\dom L$ (all subdomains $\dom L_{f}$ contain $\bot_{\dom L}$),
then we define the domain $\dom R=\dom L\multimap\dom D$ as the domain
of record functions from $\dom L$ to $\dom D$, where the universe,
$\left|\dom R\right|$, of domain $\dom R$ is defined by the equation
\begin{equation}
\left|\dom R\right|=\{\bot_{\dom R}\}\cup\bigcup_{\dom L_{f}\Subset\dom L}R(\dom L_{f},\dom D)\label{eq:rec-univ}
\end{equation}
with sets $R(\dom L_{f},\dom D)$ defined as
\begin{equation}
R(\dom L_{f},\dom D)=\{tag(\left|\dom L_{f}\right|\backslash\{\bot_{\dom L}\})\}\times\left|\dom L_{f}\strfunarr\dom D\right|\label{eq:rec-Rf}
\end{equation}
and where $tag$ is a function that maps the shape corresponding to
a domain $\dom L_{f}$ to a unique tag in a countable set of tags
(whose exact format does not need to be specified), and where $\dom L_{f}\strfunarr\dom D$
is the standard domain of strict continuous functions from $\dom L_{f}$ into
$\dom D$. Tags are needed in record functions to ensure that the
records domain constructor is a continuous, in fact \emph{computable}, domain
constructor.

To illustrate, using $\multimap$ a record $\mathsf{r}=\{l_{1}\mapsto d_{1},\cdots,l_{k}\mapsto d_{k}\}$
is modeled by a record function 
$
r=(tag(\{l_{1},\cdots,l_{k}\}),\{(\bot_{\dom L},\bot_{\dom D}),(l_{1},d_{1}),\cdots,(l_{k},d_{k})\})
$. It should be noted that $\multimap$ allows constructing the (unique)
record function $$(tag(\{\}),\{(\bot{}_{\dom L},\bot_{\dom D})\})$$
that models the empty record (one with an empty set of labels, for
which $\left|\dom L_{f}\right|=\{\bot_{\dom L}\}$.)

The approximation ordering, $\sqsubseteq_{\dom R}$, over elements
of $\dom R$ is defined as follows. The bottom element $\bot_{\dom R}$
approximates all elements of the domain $\dom R$. Non-bottom elements
$r$ and $r'$ in $\dom R$ with unequal tags are unrelated to one
another. On the other hand, elements $r$ and $r'$ with the same
tag are ordered by their embedded functions (which must be elements
of the same domain.) Formally, for two non-bottom record functions
$r,r'$ in $\dom R$ that are defined over the \emph{same }$\dom L_{f}$,
where $\left|\dom L_{f}\right|=\{\bot_{\dom L},l_{1},\cdots,l_{k}\}$,
if
\[
r=(tag(\{l_{1},\cdots,l_{k}\}),\{(\bot_{\dom L},\bot_{\dom D}),(l_{1},d_{1}),\cdots,(l_{k},d_{k})\})
\]
and
\[
r'=(tag(\{l_{1},\cdots,l_{k}\}),\{(\bot_{\dom L},\bot_{\dom D}),(l_{1},d'_{1}),\cdots,(l_{k},d'_{k})\})
\]
where $d_{1},\cdots,d_{k}$ and $d'_{1},\cdots,d'_{k}$ are elements
in $\dom D$, then we define
\[
r\sqsubseteq_{\dom R}r'\Leftrightarrow\forall_{i\leq k}(d_{i}\sqsubseteq_{\dom D}d'_{i})\label{eq:rec-approx}
\]

Having defined the records domain constructor $\rec$, we now discuss
its mathematical properties. \begin{theorem}
Given a flat countable domain of labels $\dom L$ and an arbitrary
domain $\dom D$, $\dom L\multimap\dom D$ is a domain.
\end{theorem}
\begin{proof}
See Appendix~\ref{sec:Proofs}.
\end{proof}

Because in the construction of $\NOOP$ we use $\multimap$ to construct
domains as least fixed points of functions over domains, where the
constructed domains need to be subdomains of Scott's universal domain,
$\dom U$, we need to ascertain that $\multimap$ has the domain-theoretic
properties needed for it to be used inside these functions. We thus
need to prove that $\multimap$ is a continuous function over its
input domain $\dom D$, \emph{i.e}., that, as a function over domains,
$\multimap$ is monotonic with respect to the subdomain relation,
$\Subset$, and that $\multimap$ preserves least upper bounds of domains under
that relation.

\begin{theorem}
Domain constructor $\multimap$ is a continuous function over flat domains $\dom L$ and
arbitrary domains $\dom D$.
\end{theorem}
\begin{proof}
See Appendix~\ref{sec:Proofs}.
\end{proof}

\section{\label{sec:Signatures}Class Signatures}

In this section we present formal definitions for class signatures
and related constructs. Class signatures and other signature constructs
are syntactic constructs that capture nominal information found in
objects of mainstream OO software. Embedding class signature closures
(formally defined below) in objects of $\NOOP$ makes them nominal
objects, thereby making $\NOOP$ objects more precise models of objects
in mainstream OO languages such as Java, C\#, C++, and Scala.

Class signatures formalize the notion
of object interfaces. A class
signature corresponding to a class in nominally-typed mainstream OOP
is a concrete expression the interface of the class, \emph{i.e.}, of how instances of the class should be viewed and
interacted with by other objects (``the outside world'').\footnote{Object interfaces are also discussed in~\cite{AbdelGawad2015}, ~\cite{OOPOverview13} and Ch.~2 of~\cite{NOOPbook}.}

To capture nominal information of nominally-typed mainstream OOP,
we define three syntactic signature constructs: (1) class signatures,
(2) class signature environments, and (3) class signature closures.
Additionally, fields and methods, respectively, have (4) field signatures
and (5) method signatures.

\subsection{Class Signatures}

If $\mathsf{N}$ is the set of all class names, and $\mathsf{L}$
is the set of all member (\emph{i.e.}, field and method) names, we
define a set $\mathsf{S}$ that includes all class signatures by the
equation 
\begin{equation}
\mathsf{S}=\mathsf{N}\times\mathsf{N}^{*}\times\mathsf{FS^{*}}\times\mathsf{MS^{*}}\label{eq:signature}
\end{equation}
where $\times$ and $^{*}$ are the cross-product and finite-sequences
set constructors, respectively, $\mathsf{FS}=\mathsf{L}\times\mathsf{N}$
is the set of field signatures, and $\mathsf{MS}=\mathsf{L}\times\mathsf{N}^{*}\times\mathsf{N}$
is the set of method signatures.

The equation for $\mathsf{S}$ expresses that a \emph{class signature}
corresponding to a certain class is composed of four components:
\begin{enumerate}
\item The class name (also used as a \emph{signature name} for the class
signature),
\item A finite sequence of names of \emph{immediate supersignatures} of
the signature, \emph{i.e.}, of signatures corresponding
to immediate superclasses of the class,
\item A finite sequence of field signatures corresponding to class fields, and
\item A finite sequence of method signatures corresponding to class methods.
\end{enumerate}
The use of signature names (members of $\mathsf{N}$) inside signatures
characterizes class signatures as nominal constructs, where two signatures
with different names but that are otherwise equal are different signatures.

The second component of a signature, a (possibly empty) sequence of
signature names (\emph{i.e.}, a member of $\mathsf{N}^{*}$), is the
\emph{immediate supersignature names} component of the class signature.
Having names of immediate supersignatures of a class signature explicitly\emph{
}included as a component of the class signature is an essential and
critical feature in the modeling of nominal subtyping in nominally-typed
OOP. Explicitly specifying the supersignatures of a class signature
identifies the nominal structure of the class hierarchy immediately
above the named class. This also agrees with the inheritance of the
contract associated with class names, which is a crucial semantic
component of what is intended to be inherited in nominally-typed mainstream
OOP.

The equation for field signatures expresses that a \emph{field signature}
is a pair of a field name (a member of $\mathsf{L}$) and a class
signature name. Similarly, the equation for method signatures expresses
that a \emph{method signature} is a triple of a method name, a sequence
of class signature names (for the method parameters), and a signature
name (for the method result).

Not all members of set $\mathsf{S}$ are class signatures. To agree
with our intuitions about describing the interfaces of classes and
their instances, a member $s$ of $\mathsf{S}$ is a class signature
if its supersignature names component, its field signatures component
and its method signatures component (\emph{i.e.}, the second, third
and fourth components of $s$) have no duplicate signature names,
field names, and method names, respectively (For simplicity, method
overloading is not modeled in our model of OOP.) It should be noted,
however, that field names and method names are in separate name spaces
and thus we allow a field and a method to have the same name.

Information in class signatures is derived from the text of classes
of OO programs. Given that interfaces of objects are the basis for
defining types in OO type systems, class signatures are the formal
basis for nominally-typed OO type systems, so as to confirm that objects
are used consistently and properly within a program (\cite{AbdelGawad2015}, Ch.~2 of~\cite{NOOPbook},
 and~\cite{OOPOverview13}, give more details on types and typing in OOP.)

\subsection{\label{sub:Signature-Environments}Signature Environments}

A \emph{signature environment }is a finite set of class signatures
that has unique class names, where each signature name is associated
with exactly one class signature in the environment. %
(Accordingly, function application notation can be used to refer to
particular class signatures in a signature environment. If $nm$ is
a signature name guaranteed to be the name of some class signature
in a signature environment $se$, we use function application notation,
$se(nm)$, to refer to this particular class signature.) %
 In addition to requiring the uniqueness of signature names, a finite
set of class signatures needs to satisfy certain consistency conditions
to function as a signature environment. A signature environment specifies
two relations between signature names: an immediate supersignature
relation and a direct-reference (adjacency) relation (The first relation
is a subset of the second.) These two relations can be represented
as directed graphs. The consistency conditions on a signature environment
constrain these two relations and their corresponding graphs.

As such, a finite set $se$ of class signatures is a signature environment
if and only if 
(i) \label{enu:se-ref-closed}A class signature, with the right signature
name, belongs to $se$ for each signature reference in each class
signature of $se$,
(ii) \label{enu:se-ss-no-cycles}The graph for the supersignatures relation
for $se$ is an acyclic graph %
(This constraint forces any signature environment to have at least
one class signature that has no supersignatures, \emph{i.e.}, its
second component is the empty sequence)%
, and
(iii) \label{enu:se-ss-mem-inclusion}The set of field signatures and
method signatures of each class signature $s$ in $se$ is a superset
of the set of field signatures and method signatures of each supersignature
named by the supersignatures component of $s$.

In agreement with inheritance in mainstream OO languages, the last
condition makes class signatures in signature environments reflect
the explicit inheritance information in class-based OOP, by requiring
a class signature to only extend (\emph{i.e.}, add to) the set of
members supported by an explicitly-specified supersignature. Requiring
the members of a class signature to be a superset of the members of
all of its supersignatures means that \emph{exact} matching of member
signatures is required. This requirement thus enforces an \emph{invariant
subtyping} rule for field and method signatures, mimicking the rule
used in mainstream OO languages (such as Java and C\#) before the
addition of generics. This condition can be relaxed but we do not
do so in this paper. More details are available in~\cite{NOOP}.

\subsection{\label{sub:Signature-Closures}Signature Closures}

Inside a class signature, class names can be viewed as ``pointers''
that refer to other class signatures. Without bindings of class names
to corresponding class signatures, a single class signature that has
name references to other class signatures is not a closed entity on
its own. This motivates the notion of a signature closure. A closure
of a class signature is a set of class signatures (a signature environment,
in particular) that offers bindings to class names referred to in
all elements of the set, such that the whole set has no ``dangling
pointers'' in its references to other class signatures (\emph{i.e.},
is referentially-closed) and has no redundant class signatures relative
to some main class signature in the set (called the root class signature
of the closure.) A signature closure\emph{ }thus ``closes'' the
root class signature by providing bindings for all class names referenced,
directly or indirectly, in the signature. This motivates the following
formal definition of signature closures.

A \emph{signature closure} is a pair of a signature name and a signature
environment. A pair $sc=(nm,se)$ of a signature name $nm$ and a
signature environment $se$ is a signature closure if and only if
there exists a class signature $s$ in $se$ with signature name $nm$
and\emph{ }if the direct-reference (adjacency) relation corresponding
to $se$ is referentially-closed relative to $s$, and if this relation
is the smallest such relation. Class signature $s$ is then called
the \emph{root class signature} of $sc$. Relative to the root class
signature, a signature environment is minimal, \emph{i.e.}, contains
no unnecessary class signatures. This minimality condition ensures
that all class signatures in the signature environment of a signature
closure are accessible via paths in the adjacency graph of the signature
environment starting from (the node in the graph corresponding to)
the root signature name, \emph{i.e.}, that the signature environment
has no redundant class signatures unnecessary for the root class signature.

Similar to a single class signature, when viewed as a ``closed class
signature'' a signature closure has a name: namely,\emph{ }that of
its root class signature; has member signatures: namely, field and
method signatures of its root class signature; has a fields shape
and a methods shape: namely, those of its root class signature; and
it has immediate supersignature names: namely, those of its root class
signature. A signature closure, not just a class signature, is the
full formal expression of the notion of object interfaces.
Each class in a class-based OOP program has a corresponding class
signature and a corresponding class signature closure. The nominal
information in a class signature closure is an invariant of all instances
of the class (including the behavioral contracts associated with class names.)

\subsection{\label{sub:Signature-Equality}Relations on Signatures}

For class signatures $\Sig 1=(\mathsf{nm_{1}},\mathsf{nms_{1}},\mathsf{fss_{1}},\mathsf{mss_{1}})$
and $\Sig 2=(\mathsf{nm_{2}},\mathsf{nms_{2}},\mathsf{fss_{2}},\mathsf{mss_{2}})$\textsf{,}
we define
$\Sig 1=\Sig 2 \Leftrightarrow (\mathsf{nm_{1}}=\mathsf{nm_{2}})\wedge(\mathsf{nms}_{1}\equiv\mathsf{nms}_{2})\wedge
  (\mathsf{fss}_{1}\equiv\mathsf{fss}_{2})\wedge\mathsf{(mss}_{1}\equiv\mathsf{mss}_{2})$
where $\equiv$ is an equivalence relation on sequences that ignores
the order (and repetitions) of elements of a sequence. For two field signatures \textsf{$\mathsf{fs_{1}}=(\mathsf{a_{1}},\mathsf{nm_{1}})$}
and $\mathsf{fs_{2}}=(\mathsf{a_{2}},\mathsf{nm_{2}})$, \textsf{
$\mathsf{fs_{1}}=\mathsf{fs_{2}}\Leftrightarrow(\mathsf{a_{1}}=\mathsf{a_{2}})\wedge(\mathsf{\mathsf{nm_{1}}=nm_{2}}).$
} Similarly, for two method signatures \textsf{$\mathsf{ms_{1}}=(\mathsf{b_{1}},\mathsf{nms_{1}},\mathsf{nm_{1}})$}
and $\mathsf{ms_{2}}=(\mathsf{b_{2}},\mathsf{nms_{2}},\mathsf{nm_{2}})$,
\textsf{$\mathsf{ms_{1}}=\mathsf{ms_{2}}\Leftrightarrow(\mathsf{b_{1}}=\mathsf{b_{2}})\wedge(\mathsf{nms_{1}}=\mathsf{nms{}_{2}})\wedge(\mathsf{\mathsf{nm_{1}}=nm_{2}})$} (Here, sequence equality, not sequence equivalence, is used. For
method parameter signature names, order and repetitions do matter.)

Two signature environments are equal if and only if they are equal
as sets. Two signature closures are equal if and only if they are
equal as pairs. Equal signature closures have the same root class
signature name and equal signature environments.

\textbf{}Finally, a relation between signature environments that
is needed when we discuss inheritance is the extension relation on
signature environments. A signature environment $se_{2}$ extends
a signature environment $se_{1}$ (written $se_{2}\ext se_{1}$) if
$se_{2}$ binds the names defined in $se_{1}$ to exactly the same
class signatures as $se_{1}$ does. Viewed as sets, $se_{2}$ is a
superset of $se_{1}$. Thus,
\[
se_{2}\ext se_{1}\Leftrightarrow se_{2}\supseteq se_{1}.
\]

\subsection{\label{sub:Inheritance-and-Subsigning}Subsigning and Inheritance }

The supersignatures component of class signatures defines an ordering
relation between signature closures. We call this relation between
signature closures \emph{subsigning}. The subsigning relation between
class signature closures models the inheritance relation between classes
in class-based OOP.

A signature closure $sc_{2}=(nm_{2},se_{2})$ is an \emph{immediate
subsignature} ($\subsign_{1}$) of a signature closure $sc_{1}=(nm_{1},se_{1})$
if the signature environment (\emph{i.e.}, the second component) of
$sc_{2}$ is an extension\emph{ }($\ext$) of the signature environment
of $sc_{1}$ and the signature name of $sc_{1}$ is a member of the
supersignature names component of the root class signature of $sc_{2}$,
\emph{i.e.},
\[
sc_{2}\subsign_{1}sc_{1}\Leftrightarrow se_{2}\ext se_{1}\wedge(nm_{1}\in super\_sigs(se_{2}(nm_{2}))).
\]

The subsigning relation, $\subsign$, between signature closures is
the reflexive transitive closure of the immediate subsigning relation
($\subsign_{1}$).
To illustrate the definitions given in this section,
Appendix~\ref{sec:Signature-Examples} presents a few examples of
signature constructs, and presents examples of signature closures
that are in the subsigning relation.

The inclusion of class contracts in deciding the subsigning relation
makes the subsigning relation a more accurate reflection of a true
``is-a'' (substitutability) relationship than the structural subtyping
relation used in structurally-typed OOP. This makes subsigning capture
the fact that subtyping in nominally-typed OOP is more semantically
accurate than structural subtyping, as mentioned earlier, and as is explained in more detail in~\cite{AbdelGawad2015}.

\section{\label{sec:NOOP}$\NOOP$: A Model of Nominal OOP}

Using the records domain constructor ($\rec$) presented in Section~\ref{sec:Rec}
and signature constructs presented in Section~\ref{sec:Signatures},
in this section we now present the construction of $\NOOP$ as a more
precise model of nominally-typed mainstream OOP.

The construction of $\NOOP$ proceeds in two steps. First, the
solution of a simple recursive domain equation defines a preliminary
domain $\hat{\dom O}$ of raw objects, where an object in $\hat{\dom O}$
contains (1) a signature closure that encodes nominal information
of nominally-typed OOP, and contains bindings for object members in
two separate records: (2) a record for fields of the object, and (3)
a record for methods of the object.

A simple recursive definition of objects with signature information
does not force signature information embedded in objects to conform
with their member bindings. Accordingly, in the second step of the
construction of $\NOOP$, invalid objects in the constructed preliminary
domain of objects $\hat{\dom O}$ are ``filtered out'' producing
a domain $\dom O$ of proper objects that model nominal objects of
mainstream OO software. Invalid objects are ones where the signature
information is inconsistent with member bindings in the member records.
The filtering of the preliminary domain is done by defining a projection
function on the preliminary domain $\hat{\dom O}$.

We call the model having the preliminary domain defined by the domain
equation `$pre\NOOP$'. Our target model, $\NOOP$, is the one containing
the image domain resulting from applying the filtering function on
the preliminary domain $\hat{\dom O}$ of $pre\NOOP$.

\subsection{Construction of $\NOOP$}

\label{sec:NOOP-Domain-Equation}The domain equation defining $pre\NOOP$,
and thence $\NOOP$, uses two flat domains $\dom L$ and $\mathcal{S}$.
Domain $\dom L$ is the flat domain of labels, and domain $\mathcal{S}$
is the flat domain of signature closures (Section~\ref{sec:Signatures}).

The domain equation that describes $pre\NOOP$ is
\begin{equation}
\hat{\dom O}=\mathcal{S}\times(\mathcal{L}\multimap\hat{\mathcal{O}})\times(\mathcal{L}\multimap(\hat{\mathcal{O}}^{*}\strfunarr\hat{\mathcal{O}}))\label{eq:NOOP-Domain-Equation}
\end{equation}
where the main domain defined by the equation, $\hat{\dom O}$, is
the domain of raw objects, $\times$ is the strict product domain
constructor, and $\rec$ is the records domain constructor (Section~\ref{sec:Rec}).
Equation~(\ref{eq:NOOP-Domain-Equation}) states that every raw object
(\emph{i.e.}, every element in $\hat{\dom O}$) is a triple of:
\begin{enumerate}
\item A signature closure (\emph{i.e.}, a member of $\mathcal{S}$),
\item A fields record (\emph{i.e.}, a member of $\mathcal{L}\multimap\hat{\dom O}$),
and
\item A methods record (\emph{i.e.},\emph{ }a member of $\mathcal{L}\multimap(\hat{\dom O}^{*}\strfunarr\hat{\dom O}),$
where $\strfunarr$ is the strict continuous functions domain constructor,
and $^{*}$ is the finite-sequences domain constructor.)
\end{enumerate}
Domain $\hat{\dom O}$ of $pre\NOOP$ is the solution of Equation~(\ref{eq:NOOP-Domain-Equation}).
Applying the iterative least-fixed point (LFP) construction method
from domain theory~\cite{DomTheoryIntro},
the construction of $\hat{\dom O}$
proceeds in iterations, driven by the structure\emph{ }of the right-hand
side (RHS) of Equation~(\ref{eq:NOOP-Domain-Equation}). The RHS
of the equation is viewed as a continuous function over domains (given
the continuity of all used domain constructors, and that constructor
composition preserves continuity.) Details of the iterative construction
of $pre\NOOP$ are presented in~\cite{NOOP}.

\label{sub:Filtering-and-NOOP}The second step in constructing $\NOOP$ is the
definition of a projection/fi{}ltering function, \code{filter}, to
map domain $\hat{\dom O}$ of $pre\NOOP$ to the $\NOOP$ domain $\dom O$
of valid objects modeling objects of nominally-typed OOP. For this,
first, we define an object in $\hat{\dom O}$ to be valid as follows.

\begin{definition}
\label{Def:valid-obj}An object $o$ in $\hat{\dom O}$ is \emph{valid}
if it is the bottom object $\bot_{\dom O}$, or if it is a non-bottom
object $o=(sc,fr,mr)$ such that
\begin{itemize}
\item The fields shape and the methods shape of $sc$ are exactly the same
as (\emph{i.e.}, equal to) the shape of $fr$ and the shape of $mr$, respectively,
\item Non-bottom valid objects bound to field names in $fr$ have signature
closures that subsign the signature closures for corresponding fields
in $sc$, and
\item Non-bottom functions bound to method names in $mr$ conform to corresponding
method signatures in $sc$, where by \emph{conformance} the functions are required to
\begin{itemize}
\item take in sequences of valid objects whose embedded signature closures
subsign (component-wise) the corresponding sequences of method parameter
signature closures in $sc$, prepended with $sc$ itself (for the
implicit parameter \code{self/this}), and
\item return valid objects
with signature closures that subsign the corresponding return value
signature closures specified in the method signatures in $sc$.
\end{itemize}
\end{itemize}
\end{definition}
As a direct translation of Definition~\ref{Def:valid-obj}, the function
\code{filter} mapping $\hat{\dom O}$ into $\hat{\dom O}$
($\dom O$ is a proper subdomain of $\hat{\dom O}$) is defined using
the following three recursive function definitions, presented using \emph{lazy}
functional language pseudo-code.%
\begin{lyxcode}
\textbf{fun}~filter(o:$\hat{\dom O}$):$\dom O$

\textbf{  ~match}~o~\textbf{with}~((nm,se),~fr,~mr)

\textbf{  ~if} (sf-shp(se(nm))~!=~rec-shp(fr))~$\vee$

~~~~~(sm-shp(se(nm))~!=~rec-shp(mr))

\textbf{     ~~~~return}~$\bot_{\dom O}$~//~non-matching~shapes

\textbf{  ~else}~//~lazily~construct~closest~valid object~to~o

\textbf{    ~~~~match}~se(nm),~fr,~mr~\textbf{with}

     ~~~~~~ (\_,~\_,~{[}($a_{i}$,~$snm_{i}$)~|~i=1,$\cdots$,m~{]},

             ~~~~~~~~~~~~~~{[}($b_{j}$,~$mi\_snm_{j}$,~$mo\_snm_{j}$) |~j=1,$\cdots$,n{]}),

      ~~~~~~~(fr-tag,~\{$a_{i}$~$\mapsto$~$o_{i}$~|~i=1,$\cdots$,m\}),\\
      ~~~~~~~(mr-tag,~\{$b_{j}$~$\mapsto$~$m_{j}$~|~j=1,$\cdots$,n\})

\textbf{ ~~~~let}~si~=~se\_clos(se,~$snm_{i}$)

      \textbf{ ~~~~let}~misj~=~map(se\_clos(se), {[}nm::$mi\_snm_{j}${]})

       ~~~~~~~~~//~nm~is~prepended~to~$mi\_snm_{j}$~to~ handle~`this'

      \textbf{ ~~~~let}~mosj~=~se\_clos(se,~$mo\_snm_{j}$)

        \textbf{~~~~~return}~((nm,se),

          ~~~~~~~(fr-tag,~\{$a_{i}$~$\mapsto$~filter-obj-sig(si,$o_{i}$)~|~i=1,$\cdots$,m\}),

          ~~~~~~~(mr-tag,~\{$b_{j}$~$\mapsto$~filter-meth-sig(misj,~mosj,~$m_{j}$)

          ~~~~~~~~~~~~~~~~|~j=1,$\cdots$,n\}))~\\
~\\

\textbf{fun}~filter-obj-sig(ss:$\dom S$,~o:$\hat{\dom O}$):$\dom O$

  \textbf{ ~match}~o~\textbf{with}~(s,~\_,~\_)

    \textbf{ ~if} (s~$\trianglelefteq$~ss)

      \textbf{ ~~~~return}~filter(o)~//~closest~valid~object~to~o

   \textbf{ ~else}

      \textbf{ ~~~~return}~$\bot_{\dom O}$~//~no~subsigning~\\
~\\

\textbf{fun}~filter-meth-sig(in\_s:$\dom S^{+}$, out\_s:$\dom S$, m:$\hat{\dom M}$):$\dom M$

\textbf{  ~return} ($\lambda os.$\emph{let~}vos~=~map2(filter-obj-sig, in\_s,~$os$)

\emph{ ~~~~~~~~~~~~~in}~filter-obj-sig(out\_s, m(vos)))
\end{lyxcode}
In the definition of \code{filter}, functions \code{sf-shp} and
\code{sm-shp} compute field and method shapes of signatures, while
function \code{rec-shp} computes shapes of records.\textbf{ }Function
\code{se\_clos(se,nm)} computes a signature closure corresponding
to signature name \code{nm} whose first component is \code{nm} and
whose second component is the minimal subset of signature environment
\code{se} that makes \code{se\_clos(se,nm)} a signature closure.
To handle \code{this/self} a ``curried'' version of \code{se\_clos}
is passed to the \code{map} function. Additionally, domain $\dom S^{+}$
is the domain of non-empty sequences of signature closures (non-empty
because methods are always passed in the object \code{this/self}),
and domains $\hat{\dom M}$ and $\dom M$ are auxiliary domains of
raw methods and methods, respectively. The function \texttt{\code{\texttt{map2}}
}is the two-dimensional version of \code{map} (\emph{i.e.}, takes
a binary function and two input lists as its arguments.)

In words, the definition of the filtering function \code{filter}
states that the function takes an object $o$ of $\hat{\dom O}$ and
returns a corresponding valid object of $\dom O$. If the object is
invalid because of non-matching shapes in the signature closure of
$o$ and its member records, \code{filter} returns the bottom object
$\bot_{\dom O}$ (in domain $\hat{\dom O}$, $\bot_{\dom O}$ is the
closest valid object to an invalid object with non-equal shapes in
its signature and records.) Otherwise, $o$ has matching signature
and record shapes but may have objects bound to its fields, or taken
in or returned by its methods, whose signature closure does not subsign
the corresponding signature closures in the signature closure of $o$.
In this case, \code{filter} lazily constructs and returns the closest
valid object in domain $\hat{\dom O}$ to $o$, where all non-bottom
fields and non-bottom methods of $o$ are guaranteed (via functions
\code{filter-obj-sig} and \code{filter-meth-sig}, respectively)
to have signature closures that subsign the corresponding signature
closures in the signature closure of $o$.

Function \code{filter-obj-sig} checks if its input object \code{o}
has a signature closure \code{s} that subsigns a required declared
signature closure \code{ss}. If \code{s} is not a subsignature of
\code{ss}, \code{filter-obj-sig} returns $\bot_{\dom O}$. If it
is, the function calls \code{filter} on \code{o}, thereby returning
the closest valid object to \code{o}.

For methods, when \code{filter-meth-sig} is applied to a method \code{m}
it returns a valid method that when applied to the same input $os\in\hat{\dom O}^{+}$
as \code{m}, returns the closest valid object to the output object
of \code{m} that subsigns the declared output signature closure \code{out\_s}
corresponding to the sequence of valid objects closest (component-wise)
to $os$ that (again, component-wise) subsigns the declared sequence
of input signature closures \code{in\_s} prepended with the signature
closure of the object enclosing \code{m} (to properly filter the
first argument object in $os$, which is the value for \code{this/self}.)

Having defined the filtering function \code{filter}, the proof that domain $\dom O$, as defined by \code{filter}, is
a well-defined computable subdomain of $\hat{\dom O}$ is presented in Appendix~\ref{sec:Proofs}.

\subsection{Class Types}

As constructed, $\NOOP$ is a nominal model of OOP, because objects
of domain $\dom O$ of $\NOOP$ include signatures specifying the
associated class contracts maintained by the objects (including inherited
contracts.) This nominal information encoded in signatures provides
a framework for naturally partitioning the domain of $\NOOP$ objects
into sets defining class types, where a type is a set of similar objects.

First, we define exact class types. The \emph{exact class type} corresponding
to a class \code{C} is the set of all objects tagged with the signature
closure for \code{C}.%
\footnote{In Java, for example, objects in the exact type for a class \noun{\code{\noun{C}}}
are precisely those for which the \code{getClass()} method returns
the class object for \code{C}.%
} Next, it should be noted that a cardinal principle of nominally-typed
mainstream OOP is that \emph{objects from subclasses of a class \code{\emph{C}}
conform to the contract of class \code{\emph{C}}} and can be used
in place of objects constructed using class \code{C} (\emph{i.e.},
in place of objects in the exact class type of \code{C}.) Hence,
the natural type associated with class \code{C}, called the \emph{class
type} corresponding to or designated by \code{C}, consists of the
objects in class \code{C} plus the objects in \emph{all} subclasses
of class \code{C}. In nominally-typed OO languages, the class type
designated by class \code{C} is not the exact class type for \code{C}
but the union of all exact types corresponding to classes that subclass
(\emph{i.e.}, inherit from) class \code{C}, including class \code{C}
itself.

Motivated by this discussion, we define class types in $\NOOP$
as interpretations of signature closures. For a signature closure
$sc$, its interpretation $\mathbb{S}[sc]$ is a \emph{subdomain} of domain
$\dom O$, having the same underlying approximation ordering of domain
$\dom O$ and whose universe is defined by the equation
\begin{equation}
|\mathbb{S}[sc]|=\{(scs,fr,mr)\in\dom O|scs\subsign sc\}\cup\{\bot_{\dom O}\}.\label{eq:sig-semantics}
\end{equation}
In other words, the class type designated by a class is the interpretation
of the signature closure $sc$ corresponding to the class, which,
in turn, is the set of all objects in domain $\dom O$ of $\NOOP$
with a signature closure $scs$ that subsigns $sc$, or the bottom
object $\bot_{O}$. Given that subsigning in $\NOOP$ models OO inheritance,
the definition of $\NOOP$ class types is in full agreement with intuitions
of mainstream OO developers.

Having defined class types, it should be noted that a class type $\mathbb{S}[sc]$ is always a non-empty domain (\emph{i.e.}, always
has some non-bottom object) because the object $$(sc,\{a_{1}\mapsto\bot_{\dom O},\cdots,a_{m}\mapsto\bot_{\dom O}\},\{b_{1}\mapsto\bot_{\dom M},\cdots,b_{n}\mapsto\bot_{\dom M}\})$$
(where $\{a_{1},\cdots,a_{m}\}$ is the fields shape of $sc$ and
$\{b_{1},\cdots,b_{n}\}$ is the methods shape of $sc$) is always
a valid constructed object (\emph{i.e.}, is an object of domain $\hat{\dom O}$
of $pre\NOOP$ that passes filtering to domain $\dom O$ of $\NOOP$.)
This object is a member of $\mathbb{S}[sc]$ by Equation~(\ref{eq:sig-semantics}). The non-emptiness of class types is used in the proof of the identification of inheritance and subtyping.%

\subsection{Inheritance \emph{is} Subtyping}

After we constructed $\NOOP$, and after we defined class types in agreement
with intuitions of mainstream OO developers, we can now easily see what
it means for nominally-typed OO type systems to completely
identify inheritance and subtyping. We express this statement formally
as follows: Two signature closures corresponding to two classes are
in the subsigning relation if and only if\emph{ }the class types denoted
by the two signature closures are in the subset relation (\emph{i.e.},
the two classes are in the inheritance relation if and only if the
corresponding class types are in the nominal subtyping relation.)
We prove the correspondence between inheritance and subtyping in the
following theorem.
\begin{theorem}
\label{Theorem:subsigning=00003Dsubtyping}For two signature closures
$sc_{1}$ and $sc_{2}$ denoting class types $\mathbb{S}[sc_{1}]$
and $\mathbb{S}[sc_{2}]$, we have 
\begin{equation}
sc_{1}\subsign sc_{2}\Leftrightarrow\mathbb{S}[sc_{1}]\subseteq\mathbb{S}[sc_{2}]\label{eq:inheritance=00003Dsubtyping}
\end{equation}
\end{theorem}
\begin{proof}
Based on Equation~(\ref{eq:sig-semantics}), and the non-emptiness
of class types%
, the proof of this theorem is simple.

\textbf{Case:} The $\Rightarrow$ (only if) direction:

If $sc_{1}\subsign sc_{2}$, by applying the definition of $\mathbb{S}[sc_{2}]$
(\emph{i.e.}, Equation~(\ref{eq:sig-semantics})) all elements of $\mathbb{S}[sc_{1}]$
belong to $\mathbb{S}[sc_{2}]$ (the variable $scs$ in Equation~(\ref{eq:sig-semantics})
is instantiated to $sc_{1}$, and $\bot_{\dom O}$ is a common member
in all class types.) Thus, $\mathbb{S}[sc_{1}]\subseteq\mathbb{S}[sc_{2}]$.

\textbf{Case:} The $\Leftarrow$ (if) direction:

By the non-emptiness of $\mathbb{S}[sc_{1}]$ there exists a non-bottom
object $o$ of $\mathbb{S}[sc_{1}]$ with signature closure $sc_{1}$.
If $\mathbb{S}[sc_{1}]\subseteq\mathbb{S}[sc_{2}]$, then $o\in\mathbb{S}[sc_{2}]$.
By Equation~(\ref{eq:sig-semantics}) all non-bottom members of $\mathbb{S}[sc_{2}]$
must have a signature closure that subsigns $sc_{2}$. When applied
to $o$ we thus have $sc_{1}\subsign sc_{2}$.
\end{proof}
We should notice in the proof above that it is the nominality of objects
of $\NOOP$ (\emph{i.e.}, the embedding of signature closures into
objects) that makes $\mathbb{S}[sc_{2}]$ being a superset of $\mathbb{S}[sc_{1}]$
imply that $sc_{1}$ has $sc_{2}$ as one of its supersignatures,
and vice versa. The simplicity of the proof is a clear indication
of the naturalness of the definitions for class signatures and class
types.

\section{\label{sec:NOOPvsSOOP}$\NOOP$ Compared to Structural Models of OOP}
Having presented $\NOOP$, in this section we briefly compare $\NOOP$ to the most well-known structural domain-theoretic
models of OOP, namely the model of Cardelli, which we call $\SOOP$, and that of Cook, which we call $\mu\SOOP$. 

Comparing $\NOOP$ to $\SOOP$ and $\mu\SOOP$ reveals that $\NOOP$ includes full class name information
while $\SOOP$ and $\mu\SOOP$ totally ignore this information, based on the different views of type
names adopted by each of the models.  Objects in $\SOOP$ and $\mu\SOOP$ are viewed as mere (plain) records, while in $\NOOP$
they are viewed as records that maintain contracts, which are referred to via nominal information, with nominal information being part of the identity of $\NOOP$ objects.

$\NOOP$, $\SOOP$ and $\mu\SOOP$ also have different
views of types, type inheritance and subtyping, where behavioral contracts (via type name information) are
part of the identity of types in $\NOOP$, and thus are respected in type inheritance and subtyping, but
contracts are ignored in $\SOOP$ and $\mu\SOOP$. In addition, $\NOOP$ and $\mu\SOOP$ model recursive types,
while $\SOOP$ does not. This leads $\NOOP$ (due to nominality) and $\SOOP$ (due to lack of recursive types)
to identify type inheritance with OO subtyping while $\mu\SOOP$ breaks that identification.

More details on the differences and similarities between $\NOOP$, $\SOOP$ and $\mu\SOOP$ can be found in~\cite{AbdelGawad2016}.

\section{\label{sec:Conclusions}Conclusions}

Based on realizing the semantic value of nominal-typing, in this paper we presented $\NOOP$ as a model
of OOP that includes nominal information found in nominally-typed mainstream OO software.
The inclusion of nominal information as part of the identity of objects and class types in $\NOOP$
led us to readily prove that type inheritance, at the syntactic level, and subtyping, at the semantic
level, completely agree in nominally-typed OOP. A comparison of $\NOOP$ to structural models of OOP
revealed nominal and structural models of OOP have different views on fundamental notions of OOP.
It is necessary, we thus believe, to include nominal information
in any accurate model of nominally-typed mainstream OOP. By its inclusion
of nominal information, $\NOOP$ offers a chance to understand and
advance OOP and current OO languages based on a firmer semantic
foundation.

\section{\label{sec:Future-Work}Future Work}

One immediate possible future work that can be built on top of research presented in
this paper is to define a minimal nominally-typed OO language, \emph{e.g.}, in the spirit of \noun{FJ~\cite{FJ/FGJ}},
then, in a standard straightforward manner, give the denotational semantics
of program constructs of this language in $\NOOP$. The type safety of this 
language can then be proven using the given denotational semantics.

Generics add to the expressiveness of type systems of nominally-typed OO programming languages~\cite{Bank96,Bracha98,Corky98,GenericsFAQWebsite,bloch08,JLS14,CSharp2015,ScalaWebsite}. Another possible future work that can be built on top of $\NOOP$ is to produce
a denotational model of generic nominally-typed OOP. Such a model 
may provide a chance
for a better analysis of features of generics in nominally-typed mainstream
OO languages and thus provide a chance for suggesting improvements and
extensions to the type systems of these languages.

\section*{Acknowledgments}
The authors are thankful
to Benjamin Pierce for the feedback he offered on motivating and presenting $\NOOP$.
%

\appendix

\section{\label{sec:Signature-Examples}Class Signature Examples}

To illustrate the definitions of signature constructs given in Section~\ref{sec:Signatures},
in this appendix we present a few examples of signature constructs.
Assuming the following OO class definitions (in Java-like pseudo-code),\textbf{}
\begin{lstlisting}[basicstyle={\ttfamily},language=Java]
  class Object {
    Boolean equals(Object o){ ... }
  }

  class Boolean extends Object {
    Boolean equals(Object b){ ... }
    ... // other members of class Boolean
  }

  class Pair extends Object {
    Object first, second;
    Boolean equals(Object p){ ... }
    Pair swap(){ return new Pair(second, first); }
  }
\end{lstlisting}
we define the corresponding class signatures\medskip{}

$ObjSig=$\code{(Object, {[}{]}, {[}{]}, {[}(equals, {[}Object{]}, Boolean){]})},\medskip{}

$BoolSig=$\code{(Boolean, {[}Object{]}, ...)}, and\medskip{}

$PairSig=$\code{(Pair, {[}Object{]}, {[}(first, Object), (second, Object){]},\\
\hphantom{}~~~~~~~~~~~~{[}(equals, {[}Object{]}, Boolean), (swap, {[}{]}, Pair){]})}\medskip{}
\\
and, hence, define signature environments $ObjSigEnv=$ \code{\{$ObjSig$, $BoolSig$\}},
and\\ $PairSigEnv=$ \code{\{$ObjSig$, $BoolSig$, $PairSig$\}}, and
the signature closures $ObjSigClos=$ \code{(Object, $ObjSigEnv$)},
and $PairSigClos$ = \code{(Pair, $PairSigEnv$)}.\medskip{}

We can immediately see, using the definition of extension and the
definitions of immediate subsigning and subsigning in Section~\ref{sec:Signatures},
that $PairSigEnv\ext ObjSigEnv$, $PairSigClos\subsign_{1}ObjSigClos$,
and $PairSigClos\subsign ObjSigClos$. The last conclusion expresses
the fact that class \code{Pair} inherits from class \code{Object},
and the second to last conclusion expresses that class \code{Pair}
is an immediate subclass of class \code{Object} (The reader is encouraged
to find other similar conclusions based on the definitions of classes
\code{Object}, \code{Boolean} and \code{Pair} given above.)

\section{\label{sec:Proofs}Proofs}

In this appendix we present proofs of main theorems in this paper,
pertaining to the properties of the records domain constructor $\multimap$,
and to the filtering of $pre\NOOP$ to $\NOOP$. These proofs ascertain
the well-definedness of $\multimap$ and of the filtering, and thus
their appropriateness for being used in constructing $\NOOP$.

\subsection{\label{sec:Rec-Effective-Presentation}The Domain of Record Functions
has an Effective Presentation}

It is straightforward to confirm that $\multimap$ constructs a domain.
To prove that $\multimap$ constructs domains given an arbitrary domain
$\dom D$ and a domain $\dom L$ (with a fixed interpretation as a
flat domain of labels), we build an effective presentation of the
finite elements of $\dom L\multimap\dom D$, assuming an effective
presentation of the finite elements of $\dom D$ and $\dom L$. We
prove that these finite elements form a finitary basis of the records
domain. Since $\dom L$ has a fixed interpretation, domain constructor $\multimap$
can be considered as being parametrized only by domain $\dom D$.

Given an effective presentation $L$ of $\dom L$ where $L=\left[\bot_{\dom L},l_{1},l_{2},\cdots\right]$,
we define, for all $n\in\mathbb{N}$, the finite sequences 
\[
L_{n}=[l_{j_{1}},\cdots,l_{j_{k}}]
\]
where $0<j_{1}<\cdots<j_{k}$, and 
\begin{equation}
2n=\sum_{0<i\leq k}2^{j_{i}}.\label{eq:Ln-j-k}
\end{equation}

The size $k$, of $L_{n}$, is the number of ones in the binary expansion
of $n$, and thus $k\leq log_{2}(n+1)$ with equality only when $n$
is one less than a power of 2. $k=0$ only when $n=0$, and in this
case $L_{0}=[]$ (the empty label sequence)%
\footnote{The definition of $L_{n}$ is patterned after a similar construction
presented in Dana Scott's ``Data Types as Lattices''~\cite{DTAL}.
Unlike the case in Scott's construction, $n$ here, in the LHS
of Equation~(\ref{eq:Ln-j-k}), is doubled---\emph{i.e.}, the binary expansion
of $n$ is ``shifted left'' by one position---to guarantee $j_{i}>0$,
and thus guarantee that $l_{0}=\bot_{\dom L}$ is never an element
of $L_{n}$.%
}. It is easy to confirm that there is a one-to-one correspondence
between the set of natural numbers $\mathbb{N}$ and the set of distinct
finite label sequences $L_{n}$.

Given an effective presentation of the finite elements of $\dom D$,
$D=\left[\bot_{\dom D},d_{1},d_{2},\cdots\right],$ an effective presentation
of the finite elements of $\dom D^{k}$, the domain of (non-strict)
sequences of length $k$ ($k\geq0$) of elements of $\dom D$, is
\[
(\dom D^{k})_{\pi^{k}(n_{1},n_{2},\cdots,n_{k})}=\left[d_{n_{1}},\cdots,d_{n_{k}}\right]
\]
 where, for $k>2$, 
\[
\pi^{k}(n_{1},n_{2},\cdots,n_{k})=\pi(\pi^{k-1}(n_{1},\cdots,n_{k-1}),n_{k})
\]
$\pi{}^{k}(\cdot)$ is the one-to-one $k$-tupling function (also
called the Cantor tupling function), and 
\[
\pi(p,q)=\frac{1}{2}(p+q)(p+q+1)+q=\pi^{2}(p,q)
\]
is the one-to-one Cantor pairing function.

Now, let 
\[
f(n,m)=\{(\bot_{\dom L},\bot_{\dom D})\}\cup zip(L_{n},(\dom D^{k})_{m})
\]
where, again, $k$ is the number of ones in the binary expansion of
$n$, and 
\[
zip([l_{j_{1}},\cdots,l_{j_{k}}],[d_{n_{1}},\cdots,d_{n_{k}}])=\{(l_{j_{1}},d_{n_{1}}),\cdots,(l_{j_{k}},d_{n_{k}})\}.
\]

The sequence $\mathsf{R}=\left[r_{0},r_{1},\cdots\right]$ of the
finite elements of $\dom R$ can then be presented as $r_{0}=\bot_{\dom R}$,
and for $n,m\geq0$, 
\[
r_{\pi(n,m)+1}=(tag(L_{n}),f(n,m)).
\]

Given the decidability of the consistency ($\cdot\uparrow_{\dom D}\cdot$)
and lub ($\cdot\sqcup_{\dom D}\cdot=\cdot$) relations for finite
elements of $\dom D$, the presentation $\mathsf{R}$ of the finite
elements of $\dom R$ is effective, since, for record functions $r$
and $r'$as defined in Section~\ref{sec:Definition-of-Records},
under the approximation ordering defined by Equation~\ref{eq:rec-approx},
the consistency relation
\begin{equation}
r\uparrow_{\dom R}r'\Leftrightarrow\forall_{i\leq k}(d_{i}\uparrow_{\dom D}d'_{i})\label{eq:R-consistency}
\end{equation}
is decidable (given the finiteness of records), and the lub relation
\begin{equation}
r\sqcup_{\dom R}r'=(tag(\{l_{1},\cdots,l_{k}\}),\{(\bot_{\dom L},\bot_{\dom D}),(l_{1},d_{1}\sqcup_{\dom D}d'_{1}),\cdots,(l_{k},d_{k}\sqcup_{\dom D}d'_{k})\})\label{eq:R-lub}
\end{equation}
is recursive (handling $r=\bot_{\dom R}$ or $r'=\bot_{\dom R}$ in
the definitions of $\uparrow_{\dom R}$ and $\sqcup_{\dom R}$ is
obvious. All record functions are consistent with $\bot_{\dom R}$,
and the lub of a record function $r$ and $\bot_{\dom R}$ is $r$.)
\begin{lemma}[$\multimap$ constructs domains]
\label{lem:Rec-Fin-Bas}Under $\sqsubseteq_{\dom R}$, elements of
$\mathsf{R}$ form a finitary basis of $\dom R$.\end{lemma}
\begin{proof}
Given the countability of $\dom L$ and of the finite elements of
$\dom D$, elements of $\mathsf{R}$ are countable. A consistent pair
of elements $r,r'\in\mathsf{R}$, according to Equation~(\ref{eq:R-consistency}),
has a lub $r\sqcup_{\dom R}r'$ defined by Equation~(\ref{eq:R-lub}).
Given that $\dom D$ is a domain, the lub $d\sqcup_{\dom D}d'$ of
all consistent pairs of finite elements $d$, $d'$ in $\dom D$ exists,
thus the lub $r\sqcup_{\dom R}r'$ also exists.
\end{proof}
Lemma~\ref{lem:Rec-Fin-Bas} actually proves that $\multimap$ is
a computable function that maps a pair of a flat domain and a domain to the
corresponding record domain. The presumption is that no effective
presentation is necessary for the flat domain because distinct indices
for elements of $\dom L$ will simply mean distinct labels $l_{i}$.
If $\dom L$ is a flat countably infinite domain (which implies it
has an effective presentation) and $\dom D$ is an arbitrary domain,
then the lemma asserts that $\dom L\multimap\dom D$ is a domain with an effective presentation
that is constructible from the effective presentations for $\dom L$
and $\dom D$.

\subsection{\label{sub:Rec-Cont}Domain Constructor $\multimap$ is Continuous}
\begin{lemma}[$\multimap$ is monotonic]
\label{lem:Rec-Mono} For domains $\dom D$ and $\dom D'$, and a
flat domain of labels $\dom L$, $\dom D\Subset\dom D'\Rightarrow({\dom L}\multimap\dom D)\Subset({\dom L}\multimap\dom D')$\end{lemma}
\begin{proof}
First, we prove that $\multimap$ is monotonic with respect to the
subset relation on the universe of its input, \emph{i.e.}, that $\left|\dom D\right|\subseteq\left|\dom D'\right|\Rightarrow\left|{\dom L}\multimap\dom D\right|\subseteq\left|{\dom L}\multimap\dom D'\right|$.
Then, given that the approximation ordering on $\dom D$ (as a subdomain
of $\dom D'$) is the restriction of the approximation ordering on
$\dom D'$, we prove that the elements of $\dom L\multimap\dom D$
(as members of ${\dom L}\multimap\dom D'$) form a domain under
the approximation ordering of ${\dom L}\multimap\dom D'$, and
thus that $\dom L\multimap\dom D$ is a subdomain of $\dom L\multimap\dom D'$.

Since $\left|\dom D\right|\subseteq\left|\dom D'\right|$, then $\{d_{1},\cdots,d_{k}\}\subseteq\left|\dom D\right|\Longrightarrow\{d_{1},\cdots,d_{k}\}\subseteq\left|\dom D'\right|$.
For arbitrary $\dom L_{f}$ where $\left|\dom L_{f}\right|=\{\bot_{\dom L},l_{1},\cdots,l_{k}\}$,
we thus have
\[
f=\{(\bot_{\dom L},\bot_{\dom D}),(l_{1},d_{1}),\cdots,(l_{k},d_{k})\}\in\left|\dom L_{f}\strfunarr\dom D\right|\Longrightarrow f\in\left|\dom L_{f}\strfunarr\dom D'\right|.
\]

Thus, $\left|\dom L_{f}\strfunarr\dom D\right|\subseteq\left|\dom L_{f}\strfunarr\dom D'\right|$.
Accordingly, for sets $R(\dom L_{f},\dom D)$ (the elements of $\dom L\multimap\dom D$
with tag $tag(\dom L_{f}\backslash\{\bot_{\dom L}\})$) and $R(\dom L_{f},\dom D')$
(the elements of $\dom L\multimap\dom D'$ with tag $tag(\dom L_{f}\backslash\{\bot_{\dom L}\})$),
as defined in Equation~\ref{eq:rec-Rf} of Section~\ref{sec:Definition-of-Records},
we have $R(\dom L_{f},\dom D)\subseteq R(\dom L_{f},\dom D')$. Thus,
\[
\left(\{\bot_{\dom R}\}\cup\bigcup_{\dom L_{f}\Subset\dom L}R(\dom L_{f},\dom D)\right)\subseteq\left(\{\bot_{\dom R}\}\cup\bigcup_{\dom L_{f}\Subset\dom L}R(\dom L_{f},\dom D')\right).
\]

Thus, 
\begin{equation}
\left|{\dom L}\multimap\dom D\right|\subseteq\left|{\dom L}\multimap\dom D'\right|.\label{eq:rec-mono-subset}
\end{equation}

Next, since $\dom D$ is a subdomain of $\dom D'$ when restricted
to elements of $\dom D$, we know: (i) the approximation relation
on $\dom D$ is the approximation relation on $\dom D'$ restricted
to $\dom D$; (ii) consistent pairs of $\dom D$ are consistent pairs
in $\dom D'$; and (iii) lubs, in $\dom D$, of consistent pairs of
elements of $\dom D$ are also their lubs in $\dom D'$. Thus, for
$d_{i},d_{j}\in\dom D$, $d_{i}\sqsubseteq_{\dom D}d_{j}\Leftrightarrow d_{i}\sqsubseteq_{\dom D'}d_{j}$,
$d_{i}\uparrow_{\dom D}d_{j}\Leftrightarrow d_{i}\uparrow_{\dom D'}d_{j}$
and $d_{i}\sqcup_{\dom D}d_{j}=d_{i}\sqcup_{\dom D'}d_{j}.$

Hence, according to the definition of the approximation, consistency
and lub relations for $\multimap$ (Equations~(\ref{eq:rec-approx}),~(\ref{eq:R-consistency})
and~(\ref{eq:R-lub})), the lub, in $\dom L\multimap\dom D$, of
a consistent pair of records is also their lub in $\dom L\multimap\dom D'$.
That is, respectively, for $r,r'\in\left|\dom L\multimap\dom D\right|$,
we have 
\begin{equation}
r\sqsubseteq_{\left(\dom L\multimap\dom D\right)}r'\Leftrightarrow r\sqsubseteq_{\left(\dom L\multimap\dom D'\right)}r',\label{eq:rec-mono-approx}
\end{equation}
\begin{equation}
r\uparrow_{\left(\dom L\multimap\dom D\right)}r'\Leftrightarrow r\uparrow_{\left(\dom L\multimap\dom D'\right)}r'\label{eq:rec-mono-cons}
\end{equation}
and
\begin{equation}
r\sqcup_{\left(\dom L\multimap\dom D\right)}r'=r\sqcup_{\left(\dom L\multimap\dom D'\right)}r'.\label{eq:rec-mono-lub}
\end{equation}

From equations~(\ref{eq:rec-mono-subset}),~(\ref{eq:rec-mono-approx}),~(\ref{eq:rec-mono-cons}),~(\ref{eq:rec-mono-lub}),
and the fact that $\bot_{\dom R}$ is the bottom element of both ${\dom L}\multimap\dom D$
and ${\dom L}\multimap\dom D'$, we can conclude using Definition
6.2 in~\cite{DomTheoryIntro} that 
\[
{\dom L}\multimap\dom D\Subset{\dom L}\multimap\dom D'.
\]

\end{proof}
In addition to being monotonic, continuity of a domain constructor
asserts that the lub of domains it constructs using a chain of input
domains is the domain it constructs using the lub of the chain of
input domains (\emph{i.e.}, that, for $\multimap$, the lub $\dom D$
of a chain of input domains $\dom D_{i}$ gets mapped by $\multimap$
to the lub, say domain $\dom R$, of the chain of output domains $\dom R_{i}=\dom L\multimap\dom D_{i}$.)
\begin{lemma}[$\multimap$ preserves lubs.]
\label{lem:Rec-Lubs} For a chain of domains $\dom D_{i}$, if $\dom D=\sqcup\dom D_{i}$,
$\dom R_{i}=\dom L\multimap\dom D_{i}$, and $\dom R=\dom L\multimap\dom D$,
then $\dom R=\sqcup\dom R_{i}$.\end{lemma}
\begin{proof}
Let $\dom Q$ be the lub of the chain of domains $\dom R_{i}=\dom L\multimap\dom D_{i}$
($\dom R_{i}$'s form a chain by the monotonicity of $\multimap$.)
Domain $\dom Q$ is thus the union of domains $\dom R_{i}$, \emph{i.e.},
$\dom Q=\sqcup\dom R_{i}=\bigcup_{i}({\dom L }\multimap\dom D_{i})$.

Domain $\dom Q$ is equal to $\dom R={\dom L\ensuremath{\multimap}\ }\dom D={\ensuremath{\dom L\multimap}}\bigcup_{i}\dom D_{i}$
because each element $q$ in $\dom Q$ ($q$ is a record function)
is an element of a domain ${\dom L }\multimap\dom D_{i}$ for
some $i$. Given $\dom D_{i}$ is a subset of $\dom D=\bigcup_{i}\dom D_{i}$,
$q$ will also appear in $\dom R$.

Similarly, a record function $r$ in $\dom R$ is an element of a
domain ${\dom L }\multimap\dom D_{i}$ for some $i$, because
every \emph{finite} subset of $\bigcup_{i}\dom D_{i}$ has to appear
in \emph{one} $\dom D_{i}$ (given that $\dom D_{i}$ is a chain of
domains.) Thus, by the definition of $\dom Q$, $r$ is also a member
of $\dom Q$.

This proves that $\dom Q=\dom R$.
\end{proof}
Lemmas~\ref{lem:Rec-Mono} and~\ref{lem:Rec-Lubs} prove that $\multimap$
is computable given effective presentations for $\dom L$ and $\dom D$
(or, equivalently, an effective presentation for $\dom D$.)

\subsection{\label{sub:Filtering-is-Fin-Proj-Proof}Filtering is a Finitary Projection}

In this section we prove that function \code{filter}, as defined
in Section~\ref{sub:Filtering-and-NOOP}, is indeed a finitary
projection, and thus that the domain $\dom O$ of valid objects
(Definition~\ref{Def:valid-obj} in Section~\ref{sub:Filtering-and-NOOP})
defined by the filtering function is a subdomain of Scott's universal
domain $\dom U$, and thus is indeed a domain.

To do so, we first prove a number of auxiliary propositions regarding
domain $\hat{\dom O}$.
\begin{proposition}
\label{prop:high-rank-no-approx-low-rank}In domain $\hat{\dom O}$,
higher-ranked objects do not approximate lower-ranked ones, \emph{i.e.},
$rank(o_{1})<rank(o_{2})$ implies $o_{2}\not\sqsubseteq o_{1}$\end{proposition}
\begin{proof}
By strong induction on rank of objects.
\end{proof}
To prove that \code{filter} defines a projection, in the sequel
we use the inductively-defined predicate \code{valid} (as defined
by Definition~\ref{Def:valid-obj} in Section~\ref{sub:Filtering-and-NOOP})
that applies to objects of $\hat{\dom O}$. Note that, in addition
to $\bot_{\dom O}$, objects with empty field and method records provide
base cases for the definition of \code{valid}.
\begin{lemma}[\texttt{filter} returns the closest valid object that approximates
its input object]
\texttt{\label{lem:filter-ret-closest-valid}}For an object $o$
of $\hat{\dom O}$,\texttt{ filter($o$)$\sqsubseteq$$o$ $\wedge$
valid(filter($o$))} $\wedge$ $\forall o'$ ($o'$$\sqsubseteq$$o$
$\wedge$ \texttt{valid($o'$)} $\Longrightarrow$\texttt{ $o'$}
$\sqsubseteq$\texttt{ filter($o$)})\texttt{}\end{lemma}
\begin{proof}
By strong induction on rank of objects, noting that, for the base
case, \texttt{filter(o)} diverges (\emph{i.e.}, ``returns'' $\bot_{\dom O}$)
for the rank 0 input object $\bot_{\dom O}$, and if an object $o$
of rank 1 is invalid then \texttt{filter(o)} also returns $\bot_{\dom O}$
(no distinct objects of rank 1 approximate each other.) Proposition~\ref{prop:high-rank-no-approx-low-rank}
is used for the inductive case.\end{proof}
\begin{theorem}
\texttt{\label{thm:filter-is-fin-proj}filter} is a finitary projection.\end{theorem}
\begin{proof}
We prove that\texttt{ filter} is a finitary projection, on four steps.
\begin{enumerate}
\item \texttt{filter} is a retraction: \texttt{filter(filter(o)) = filter(o)}

\begin{proof}
Obvious from definition of \texttt{filter}, and that, by Lemma~\ref{lem:filter-ret-closest-valid},
function \texttt{filter} returns a valid object (\emph{i.e.},\texttt{
valid(filter(o))}).
\end{proof}
\item \texttt{filter} approximates identity: \texttt{filter(o) $\sqsubseteq$
o}

\begin{proof}
By Lemma~\ref{lem:filter-ret-closest-valid}.
\end{proof}
\item \texttt{filter} is a continuous function

\begin{proof}
Direct, from the continuity of functions used to define \code{filter}
(such as \code{rec-shp}, \code{map}, \code{se\_clos}, etc.), and noting
the closure of continuous functions under composition and lambda abstraction.
\end{proof}
\item \code{filter} is finitary

\begin{proof}
The condition in point 2 of Theorem 8.5 in~\cite{DomTheoryIntro},
namely 
\[
a(x)=\{y\in\dom O|\exists x'\in x.x'ax'\wedge y\sqsubseteq x'\},
\]
can be rewritten for the filtering function \code{filter} as
\begin{equation}
\mathtt{filter}(o)=\{p\in\dom O|\exists o'\in\dom O.o'\sqsubseteq o\wedge o'=\mathtt{filter}(o')\wedge p\sqsubseteq o'\}.\label{eq:filter-cond}
\end{equation}
Objects of domain $\hat{\dom O}$ are in one-to-one correspondence with principal
ideals over their finitary basis. The filtering function \code{filter}
returns, as its output, the closest valid object to its input object
(The object returned is a well-defined object, and it is a fixed point
of the filtering function.) Thus, given that objects correspond to strong ideals
in the finitary basis of $\hat{\dom O}$, they correspond to downward-closed
sets. Condition~(\ref{eq:filter-cond}) is thus true for all objects in $\hat{\dom O}$.\end{proof}
\end{enumerate}
Based on the definition of finitary projections, function \code{filter} is thus a finitary projection.\end{proof}

\bibliographystyle{authordate3}

\end{document}